\documentclass[12pt]{article}
\usepackage{amsfonts}
\usepackage{algpseudocode}
\usepackage{algorithm}
\usepackage{float}
\usepackage{graphicx}
\usepackage{amssymb}
\usepackage{amsmath}
\usepackage[
backend=biber,
style=alphabetic,
sorting=ynt
]{biblatex}
\usepackage[rightcaption]{sidecap}
\usepackage{amsthm}

\numberwithin{equation}{section}
\theoremstyle{plain}
\newtheorem{theorem}{Theorem}[section]
\newtheorem{lemma}[theorem]{Lemma}

\theoremstyle{definition}

\begin{document}
\baselineskip 13pt

\begin{center}
\textbf{\Large Pulsar Consensus} \\
\bigbreak
\noindent
\textbf{\small Version 0.1} \\

\vspace{1.5cc}
{ \sc Samer Afach, Ben Marsh, and Enrico Rubboli}\\

\vspace{0.3 cm}

{\small Mintlayer}
 \end{center}
\vspace{0.3 cm}

\begin{abstract}
  \noindent  In this paper, we informally introduce the Pulsar proof of stake consensus paper and discuss the relevant design decisions and considerations. The Pulsar protocol we propose is designed to facilitate the creation of a proof of stake sidechain for a proof of work blockchain. We present an overview of a novel composable density-based chain selection rule for proof of stake systems which can be seen as a superset of some standard existing longest chain rules for proof of stake protocols. We discuss the Pulsar protocol in comparison to existing proof of stake protocols and define its benefits over existing designs while defining the limitations of the work. Pulsar is currently implemented in the Mintlayer proof of stake Bitcoin sidechain.
\end{abstract}

\bigbreak
\noindent

\newpage
\tableofcontents
\newpage

\section{Introduction}
As cryptocurrencies, which were first introduced in their current form by Bitcoin \cite{btc} and more generically blockchains, mature, there is a need for evolution in their underlying protocols to meet the demand for robust, secure, scalable, and more energy-efficient blockchains. Constrained by the CAP \cite{cap} theorem, there are a number of trade-offs that must be considered when designing a blockchain and its consensus protocol.
\bigbreak
\noindent
By design, consensus protocols are inherently resistant to Sybil attacks, thereby safeguarding trust and integrity within the network. This is important to prevent the potential manipulation of the network by any single entity seeking control. A consensus protocol, such as Bitcoin, must be a fault tolerant system, such as those found in Byzantine Fault Tolerant protocols, characterized by its capability to operate accurately, even in scenarios where a subset of nodes becomes faulty, inaccessible, or intentionally malicious. We assume protocols such as Bitcoin are BFT hereafter. 
\bigbreak
\noindent
Proof of stake consensus protocols have far more complicated security models than traditional proof of work based protocols. The reliance on work in proof of work systems is bypassed in proof of stake systems allowing for the cheap creation of malicious blocks which can be used to attack the honest chain. The Proof of Stake protocols require a meticulous assessment of potential attack vectors such as nothing-at-stake attacks, long-range attacks, and stake grinding.
\bigbreak
\noindent
Moreover, the current development of consensus protocols has ushered in the era of dynamic availability - a concept that allows a network to maintain consensus, even when the number of participating nodes changes dynamically. This differs significantly from the traditional model of multiparty computation, where the number of participants is known and fixed beforehand. Dynamic availability offers superior flexibility and adaptability, allowing blockchain networks to continue operating reliably in the face of constant change.
\bigbreak
\noindent
In this work, we use the term sidechain \cite{sidechain} to mean a separate blockchain with a direct link through a bridge or atomic swap system allowing the chains to be used in parallel. Typically one of the chains is dominant, although we reject the idea of a master-slave relationship between the chains, and the less dominant chain, which acts for the benefit of the mainchain either by enabling new features or enhancing existing ones, is the sidechain.
\bigbreak
\noindent
Pulsar, describes a proof of stake consensus protocol designed to be used as a sidechain for a proof of work layer one blockchain but is usable in isolation as the sole consensus protocol in a network. The consensus protocol is designed to be formally proven and guarantee composability where relevant and possible. We assume a similar security model to a partially-synchronous Byzantine Fault Tolerant protocol with PKI. Formal discussions of the performance and security of Pulsar are out of the scope of this document and will follow. 
\section{Previous Work}
Proof of stake consensus protocols were first introduced by Peercoin \cite{pc, ppc} in 2012 and, strictly speaking, created a hybrid cryptocurrency that used both proof of work and proof of stake within the protocol. As a proof of stake protocol, Peercoin requires staking, and this is done by coins being unspent and left in place; after a 30-day period, the UTXOs are considered matured and stake-able. Peercoin helpfully aggregates matured UTXOs for the purpose of the consensus protocol. The use of coin age, defined as the length of time the UTXO has existed multiplied by the amount of time the UTXO has existed \textit{$coins \cdot time$}, in Peercoin has led to a handful of issues, including peculiarly disincentivizing the liveliness of an agent by disincentivizing staking and allowing for a lower barrier of entry to 51\% attacks where a much lower percentage than 51\% of stake is needed to conduct such an attack. In some Peercoin implementations, an upper bound of 90 days has been used in an attempt to mitigate some of these issues. Despite this, the use of matured UTXOs is still considered contentious at best and has been routinely avoided by later proof of stake protocols. A staker in Peercoin not only earns a block reward when they produce a block at a rate proportional to their coin age, in a similar way to modern proof of stake systems having block production rates proportional to their relative stake but also earns a small extra reward, sometimes viewed as interest on top leading to a situation where the reward for a block is proportional to the coin age. Peercoin does not punish nodes \cite{bentov} as one might be familiar with in more modern proof of stake systems; for example, nothing-at-stake attacks are primarily handled via the coin age concept; there have also been attempts to mitigate such attacks by removing the tip if the same coins produced several blocks, although this has not been without controversy.
\bigbreak
\noindent
The lack of punishment for misbehaving nodes has been controversial in some quarters, and the Ethereum foundation's Beacon Chain, previously and still sometimes known as Ethereum 2.0 \cite{eth} amongst other things, introduces slashing in their proof of stake upgrade; in essence, a node which acts against what the networks believe is valid behavior then the agent will lose some of their stake. In some ways, the idea of punishing for misbehavior can be thought of as an attempt to equalize the risk profiles of proof of work mining and proof of stake staking, as in a proof of work system, a malicious agent will spend money on electricity when attempting to attack the network, and slashing creates a similar utility for the agent. This is meant to incentivize good behavior. The slashing behavior in Ethereum's Beacon Chain means that Ethereum's Beacon Chain does not behave in the same way as Bitcoin when it comes to dynamic availability, and an agent is punished for lack of liveliness. In this sense, Ethereum is closer to more traditional distributed multiparty computation protocols.
\bigbreak
\noindent
Peercoin, unlike most modern proof of stake protocols such as Algorand \cite{algorand}, BABE \cite{polkadot}, and Cardano's Ouroboros \cite{cardano}, does not use a verifiable random function for randomness, when selecting the block producers, when required on-chain and in a previous implementation used the signature of the \textit{$N-1^{th}$} block as the source of randomness for the \textit{$N^{th}$} block. This made it possible to attack the block producers through a denial of service attack. The use of a VRF, as defined in \ref{vrf}, allows for block producers to be selected without a public list of block producers being available. 
\bigbreak
\noindent
In some proof of stake implementations, stake grinding issues have existed, allowing a block producer to increase the chance of them producing a later block. In some cases, these have been due to the use of something malleable to the block producer being used as a source of randomness, used for selecting the next block producer, such as a block hash or the signature of a block. For example, in Peercoin, parameters were left in blocks that were hashed, allowing the block producer to impact the proof of stake hash. In essence, they could grind parameter options in the block until they found the most favorable hash output to them. Some protocols have used slashing to penalize agents accused to stake grinding, and lockup periods of stakes can make stake grinding attacks more expensive, reducing the utility of the attack if not preventing it entirely. 
\bigbreak
\noindent
Finality can be a contentious issue in some circles, with some believing Bitcoin's lack of formal finality via statistical finality is the holy grail; since Pulsar is designed to be a sidechain for a proof of work chain such as Bitcoin, this is an important consideration. Most proof-of-stake consensus protocols incorporate finality in some form in order to reduce the risk of long-range attacks, but each has its own model for doing so. In Algorand, finality is done on a block-by-block basis meaning there is no chance to fork or reorganize the blockchain. Polkadot, with its many parachains, has a thin line to walk between being able to reorg and not reorganizing too deep as to cause the parachains pain, and because of the limitation of storing candidates in memory, limiting the length of candidate branches. It achieves this via its GRANDPA \cite{grandpa} finality gadget, which is capable of finalizing blocks en-masse. The limited length of candidate branches has the potential to cause issues in a situation where the consensus is proposed on a sidechain which must reorganize alongside a mainchain which may not have the same limitations. 
\bigbreak
\noindent
In proof of work systems, chain selection is a relatively simple and well-defined process, based on the assumption the cryptographic primitives they rely on, such as SHA-256 for Bitcoin, remain secure, which consists of selecting the chain with the most work; however formal barriers \cite{posbarriers} exist to the same idea in proof of stake systems. As a result, it is necessary for proof of stake systems to take a different approach. In Ouroboros, this consists of a chain selection rule which is, in essence, a pair of rules, once for recent blocks and once for blocks beyond a finality threshold. Ouroboros looks at the density of the chain, meaning the chain with the fewest empty slots, within a given time period to determine the canonical chain; selecting the denser chain makes it much more challenging to perform an attack that uses a fork, such as building on a private chain to be released in the future.
\bigbreak
\noindent
Although proof of work protocols have been formally analyzed and formal models \cite{miller} for them exist, it wasn't until Ouroboros that a similar approach had been taken with proof of stake systems. Since then, several other proof of stake protocols have undergone a formal security analysis, but a single unifying model does not exist due to different protocols' disparate aims. As a result, the work does not conform to a standard approach for analyzing a proof of stake system but draws on work from several existing protocols. 

\section{Contribution}
Pulsar builds upon the work of Ouroboros and BABE and introduces a consensus protocol that is not only formally verifiable like Ouroboros but suitable to be used as a sidechain to a proof of work blockchain, most notably Bitcoin. Pulsar does not currently punish agents for liveliness or adverse behavior, just as in a proof of stake system like Peercoin, although it is likely in a future version that an incentive for an agent to be online will be introduced. Pulsar introduces finality after a depth of 1000 blocks in order to ensure compatibility with the Bitcoin mainchain, which it is designed to work alongside; as such, Pulsar has a hybrid finality model with deterministic finality of 1000 blocks and probabilistic finality before. 
\bigbreak
\noindent
Most notably, Pulsar introduces a new composable chain selection rule where a single rule can be used to select the canonical chain from the entire history, not just until an arbitrary point in history, assuming the chain includes the required checkpoints. This chain selection rule is discussed in more detail in \ref{cs} but can be seen as a superset of Cardano's Ouroboros chain selection rule and more traditional approaches, such as the Peercoin approach. It is designed to be tune-able to allow the blockchain implementing Pulsar to get their desired behavior. Pulsar is designed with flexibility in mind to be tune-able such that it is able to be reused as a sidechain for non-Bitcoin sidechains or as a standalone blockchain. The consensus protocol has been implemented as part of the Mintlayer \cite{mintlayer} blockchain.

\section{Network and Security discussion}
\subsection{Cryptography}
\subsubsection{Hash function}
The implementation of Pulsar relies on the existence of a cryptographically secure hash function  \textit{$h: A \to B$} defined over a message alphabet \textit{$A = {a \in \{0,1\}^{i}} : i \in \mathbb{N}$} producing a digest \textit{$B = \{0,1\}^{k}$}. The hash function should have the properties of collision resistance, pre-image resistance, and an avalanche effect. BLAKE2 \cite{blake} provides a viable option, although other alternatives exist.
\subsubsection{Digital signature}
Pulsar requires a digital signature scheme for the signing of blocks and transactions within the network. Informally a signature scheme requires a public and private key pair \textit{$pk$} and \textit{$sk$} and relies on a function \textit{$F_{s}$} to sign a message, \textit{$m$}, \textit{$sig = F_{s}(sk, m)$} and a function \textit{$F_{v}$} to verify a signature using the public key \textit{$ver = F_{v}(pk, m ,sig)$}. An aggregable signature scheme such as BLS or Schnorr is preferred since it allows for more efficient protocols to be designed in a resource-constrained environment such as a blockchain. In our implementation, we use Schnorr \cite{schnorr} over BLS for the simplicity of reusing existing BIPs from Bitcoin despite the fact BLS offers more efficient aggregation. Notably, we use the curve secp256k1 as in Bitcoin due to known issues using Ed25519 with certain BIPs \cite{ed25519}.
\subsubsection{KES}
Although not essential, there is a significant benefit to the use of a key evolving signature scheme \cite{kes} for the digital signatures used to sign blocks in Pulsar. A KES is used to create a chain of signing keys preventing the ability to reuse old keys if appropriately used, as it allows for forward secrecy and back-dating of signatures. It is possible to get similar behavior in the system by ratcheting keys after a defined period of time, be that a single use or a number of days.
\subsubsection{VRF}\label{vrf}
A verifiable random function \cite{vrf} is required to select agents on-chain randomly. Using on-chain data allows the selection to be deterministic, opening the opportunity for denial-of-service attacks against known selected agents; removing the possibility for these attacks is the primary benefit of a VRF. In some circumstances, a stake-grinding-esque attack may be possible where the source of randomness can be influenced by the agents themselves, and a VRF can also prevent circumstances like these from arising, no matter how rare they are. A VRF is a cryptographic primitive that allows for the creation of a pseudo-random output and the creation of a proof of that output. To define this slightly more formally while still avoiding getting into the weeds, we can think of an EC VRF for an agent, \textit{$a$}, as consisting of a public and private key pair called \textit{$pk_{a}$} and \textit{$sk_{a}$} respectively, along with functions \textit{$f_{r}$} and \textit{$f_{v}$} where \textit{$f_{r}$} creates a pseudo-random output and a commitment based proof and \textit{$f_{v}$} verifies the pseudo-random output of \textit{$f_{r}$} along with the proof and \textit{$pk_{a}$}. 

\subsection{Networking}
\subsubsection{Liveliness and dynamic availability}
The network is assumed to be dynamically available where a node or pool running agent can enter and exit the network at will. However, an agent joining the network is not able to earn a block reward immediately as part of the security model of the network. Agents within the network are assumed to be lively and are not currently punished for lack of liveliness, although they will be so in the future.
\subsubsection{Network latency}
The network latency is bounded by some \textit{$\Delta$} such that all network participants receive a block, \textit{$B$}, by time \textit{$t + \Delta$}. Given the existence of \textit{$\Delta$} and the use of a verifiable random function, there is a probability \textit{$\epsilon_{\mathrm{cons}}$} that at a height \textit{$h$}, there are competing blocks \textit{$B$} and \textit{$B'$} such that \textit{$B \neq B'$}.

\section{Consensus}

\subsection{Blocks}
Pulsar blocks are not radically different from blocks in other blockchains, even those in Bitcoin; a block is produced by a block producer, which in Pulsar is a pool selected as the leader for a given slot. Once a block has been built, it is sent out to the rest of the network, and under normal conditions, the block will become the latest block in the canonical chain. In our implementation, a block is produced every 2 minutes, on average, with a maximum size of 1MB.
\subsubsection{Block header}
Blocks in Pulsar contain a header that includes a set of fields common to all blocks. In a proof of work system, the block header is sufficient enough to allow for header first syncing via the consensus data, but in a proof of work protocol like Pulsar, the consensus data is not alone sufficient. 
\bigbreak
\noindent
A Pulsar block header includes: a signature; consensus data such as the VRF data, pool ID, and target; the previous block ID; the timestamp; the Merkle root; and the witness Merkle root. The VRF data allows a node to validate that the pool was indeed a slot leader when the block was produced and thus was allowed to create a block. The signature is necessary to ensure the pool claimed to be responsible for the block is responsible for the block and was allowed to do so and ensures that the block cannot be tampered with later. Being able to prove the producer of a block is also essential if the network intends to punish malicious actors, which is an open area of future work.
\subsubsection{Validation}
A Pulsar block timestamp has a hard requirement to come later, in UNIX time, then the block upon which it is attempting to build. That is \textit{$T_{n} > T_{n-1}$} for any \textit{$n$} where \textit{$T_{n}$} is the timestamp of a block at height \textit{$n$}. This is required to ensure strict ordering of the blocks within the network. For the rest of the block, the validation steps are obvious, but for the sake of completeness, we will list the important ones here. There are a variety of steps that must be taken in order for a block to be validated; notably, both the header and block body have to both be checked. The process of checking the header of a block consists of checking the block parent, ensuring the maximum reorganization depth is obeyed, and ensuring that the consensus data is valid. Further checks also take place in the form of checking the size of the block, ensuring that the block reward in the block is valid, validating the Merkle tree roots, and checking the transactions; ultimately, the block is used to update the local state of the node. Errors from invalid blocks can be handled in different ways and is arguably an implementation detail that can be used to shape the behavior of the system which implements Pulsar.
\subsubsection{Block body}
The body of a block in Pulsar is not worth discussing in any acute detail, so we will just mention that it contains the transactions included in the block and the block reward, if there is one.

\subsection{Staking}
\subsubsection{Staking} \label{staking}
As a proof of stake consensus protocol, block producers are selected from agents who have staked coins in the network. More specifically, in Pulsar, a block producer is a pool who is a slot leader for a given slot.
\subsubsection{Epochs}
As with most proof of stake protocols, Pulsar breaks time into slots and epochs. In Pulsar, an epoch is 5 days in length, 3600 blocks, and is used to distribute rewards to the agents in the network at a defined interval as well as update the randomness for that epoch. Key evolving signature schemes can be tied to epochs to enforce long-range attack mitigation, although these are not enforced by the protocol. An epoch is a set of ordered second long slots, \textit{$e_{i} = sl^{i}_{1}, sl^{i}_{2}...sl^{i}_{n}$} where each slot may or may not have a block produced. Epochs themselves are also ordered and strictly non-overlapping \textit{$e_{1}, e_{2}, ..., e_{n}$}.
\subsubsection{Pools and delegation}
A pool must consist of no less than some staked threshold to exist, and this will be referred to as the pledge hereafter; the purpose of the pledge is to incentivize good behavior by enforcing a pool operator having a financial interest in the pool. In future versions of Pulsar, it is possible that punishments such as slashing will apply to a pool's pledge when the pool behaves in a provably malicious fashion. 
\bigbreak
\noindent
In our implementation, we set the minimum pledge at \textit{$0.01\%$} of the total supply. The pool operator is responsible for maintaining the node, including ensuring liveliness as well as honest behavior. A pool owner does not take control of coins staked in their pool at any time, and they remain under the control of the delegator.
\bigbreak
\noindent
Any coin holder in the network may choose to delegate their coins to any pool in order to become an active participant in the network and earn block rewards, and help secure the network. The pool has a total stake that is equal to the sum of the pledge and the sum of all the pool's delegators' stakes {\[ p_{s} = p_{p} + \sum_{d \in [D_{p}]} d \]}A pool has a relative stake which is the fraction of the total stake of the network held by that pool. {\[\phi = \frac{p_{s}}{\sum_{p \in [P]}p}\]}The probability of a pool with a relative stake \textit{$\phi$} producing a given block can be seen as \textit{$Pr[p] \propto \phi$}.
\bigbreak
\noindent
Given this, a pool can predict its expected utility in the network as \textit{$\phi \cdot r$} where \textit{r} is a quasistatic variable representing the block reward. The value of \textit{$r$} is not only influenced by the emission curve which defines the reward for a given block but the rewards formula which exists to incentivize larger pledges for pools, with a pledge \textit{$p_{1}$} twice the size of \textit{$p_{2}$} having \textit{$0.5\%$} extra power, if the minimum required pledge is \textit{$p_{\mathrm{min}}$} then pledging \textit{$2 \cdot p_{\mathrm{min}}$} will result in a single, non-linear \textit{$0.5\%$} bonus stake probability.
\bigbreak
\noindent
A delegator can find their expected utility by taking into account their relative stake within the pool \textit{$\frac{d}{p_{s}} \cdot \phi \cdot r - f_{p}$}.
\bigbreak
\noindent
In order to limit the strength of a given pool the total size of a pool is limited to \textit{$\frac{1}{1000}$} of the total supply after which it saturates. That is to say that the increased probability of producing a block approaches an asymptote as \textit{$p_{s} \to \tau_{\mathrm{sat}}$}.

\subsubsection{Pool pledge incentives}
In order to increase decentralization and incentivize pool owners to pledge, pools with a higher pledge receive a slightly higher expected return to a bound defined to encourage decentralization. That is to say pools saturate when they reach \textit{$\frac{1}{k}$} of the total final supply. In Mintlayer the value chosen is \textit{$k = 1000$}. The total supply is used to increase the simplicity of strategy finding which could be computationally expensive with dynamic reward incentives.
\[ \frac{\sigma + a \cdot (\frac{\sigma - s \cdot (\frac{z - \sigma}{z})}{z})}{a + 1} \to \sigma - \frac{a}{a + 1} \cdot (\frac{\sigma \cdot z^{2} - s \cdot (z \cdot \sigma - s \cdot (z - \sigma))}{z^{2}})\]
We should discuss several parameters now starting with the obvious \textit{$z = \frac{1}{k}$} which represents the size of a saturated pool. \textit{$a = \frac{75375728}{1000000000} = 0.07537578$} which represents a \textit{$0.5\%$} incentive increase for the first doubling of the pledge amount, that is to say a pledge of \textit{$80,000$} has a \textit{$0.5\%$} benefit over 2 pools of \textit{$40,000$}. This parameter can be tuned according to the benefit desired. \textit{$s$} represents the pledge of a given pool and \textit{$\sigma$} the total stake of that pool.

\subsubsection{Accounting}
Pulsar utilizes an account-based accounting system to handle staking irrespective of the underlying blockchain; by this, we mean that a UTXO-based blockchain such as Mintlayer still uses an account abstraction to simplify staking. An accounting abstraction allows for delegation to be implemented in a way that the pool owner has the ability to prove rights over those coins without ever taking ownership of them. The pool abstraction also limits the number of UTXOs being paid from block rewards, and you could, for example, imagine a case where entire blocks are filled with paying delegators. In contrast, through the accounting system, these UTXO-based transactions are only required when a delegator wishes to take direct ownership of their rewards. The accounting system has the ability to create pools, decommission pools, increase the pool owner's pledge, delegate, and un-stake a delegation. A delegator who unstakes or an agent who decommissions a pool has an enforced cool-down period before the funds are spendable. 

\subsection{Slot leader selection}
Any extant pool is eligible to make a block but to do so they must be considered a leader for a given slot. Recall that a slot leader is determined randomly with a probability proportional to the pool's relative stake as explained in \ref{staking}.
\subsubsection{VRF}
A slot leader is chosen at random for a given slot, \textit{sl}, if the pool's VRF output is below a threshold, \textit{$O^{vrf}_{sl} < \tau$}. Since the probability of producing a block should be proportional to a pool's relative stake \textit{$\phi$}, the threshold is multiplied by the pool's total stake \textit{$\tau \cdot p_{s}$} before the VRF output is compared to it. It is worth noting that since this selection is probabilistic and random, it is both possible and valid to have multiple slot leaders chosen for any given slot in time. This situation will be rectified by the selection of the longest chain after further blocks are built upon one of the competing blocks.
\subsubsection{Difficulty adjustment}
In order to keep the block production time as close to constant as possible, \textit{$\tau$} changes dynamically after every block produced, looking at the average block times of the previous \textit{100} blocks. The threshold swing is limited to \textit{$\frac{1}{1000} \cdot \tau_{\mathrm{prev}}$} every block acting as a low-pass filter to prevent large oscillations.
\begin{algorithm}
\caption{Threshold adjustment}\label{alg:cap}
\begin{algorithmic}
\Require $\tau_{\mathrm{prev}} \geq 0$
\State $d = \frac{\tau_{\mathrm{prev}}}{1000}$
\State $\tau_{\mathrm{new}} = \tau_{prev} \cdot \frac{t_{\mathrm{actual}}}{t_{\mathrm{target}}}$
\If{$\tau_{\mathrm{new}} < \tau_{\mathrm{prev}} - d$}
\State $\tau_{\mathrm{new}} = \tau_{\mathrm{prev}} - d$
\ElsIf{$\tau_{\mathrm{new}} > \tau_{\mathrm{prev}} + d$}
\State $\tau_{\mathrm{new}} = \tau_{\mathrm{prev}} + d$
\EndIf
\end{algorithmic}
\end{algorithm}
\subsection{Finality}
\subsubsection{Reorganization}
Pulsar is designed to work as a sidechain to a PoW-based blockchain such as Bitcoin, and as such, Pulsar does not support block-by-block finality as found in PoS-based chains like Algorand as it must be able to reorganize alongside Bitcoin if necessary. Since Bitcoin relies on statistical finality where, it can, technically, if not socially, reorg to an arbitrary depth, given the constraints of being able to finalize blocks in Pulsar and supporting arbitrary reorganize a quasi-arbitrary depth of \textit{1000} blocks has been chosen as the maximum depth of a reorg in Pulsar as a reorganization longer than approximately \textit{33} hours would be socially unacceptable. As a result, a Mintlayer block and the transactions it contains can be considered finalized after \textit{1000} blocks.
\subsubsection{Checkpointing}
Checkpointing is an essential part of the Pulsar protocol as it provides a defense against long-range attacks \cite{possidechain}. In essence, checkpointing is used to immutably enforce a canonical chain in the history before the checkpoint by hardcoding a specific blockhash into the protocol. That is that a chain that doesn't agree on the blocks earlier than the checkpoint cannot become the canonical chain, as the chain cannot diverge from that hardcoded point. 
\subsection{Chain selection}\label{cs}
In order to select the canonical chain Pulsar uses a single composable chain selection rule designed to choose a chain by its density. The rule will strictly select the chain with the densest history in the latest set of blocks because finalized blocks must be included. If we define history as a string \textit{$\{0,1\}^{n}$} then a \textit{0} can be seen as an empty slot and a \textit{1} as a slot with a valid block. The intention of the Pulsar chain selection rule is to pick the denser chain from these. Doing this directly by purely counting blocks is impossible with a single rule, as it requires looking at multiple blocks at once rather than a single block.
\bigbreak
\noindent
The longest chain rule can, in some ways, be seen as an extension of the longest chain rule found in Cardano's Ouroboros. In Ouroboros, the coefficients of the "slot string" are mapped to 1 and -1, meaning that negative outputs are valid, which is not true in the Pulsar longest chain rule. More formally we could say the codomains of the longest chain rules can be thought of as \textit{$f_{p}: \mathbb{Z}^{+} \mapsto [0, \infty) \subset \mathbb{R}$} and \textit{$f_{o}: \mathbb{Z}^{+} \mapsto \mathbb{Z}$}. The key benefit here being that the chain trust in Pulsar will never be negative irrespective of the number of unfilled slots. 
\bigbreak
\noindent
As a result, Pulsar's chain selection rule works on the basis of summing a chain trust value of valid blocks in slots and removing a unit of trust for empty slots. As a result, the total maximum trust for a given chain can be calculated since it is just \textit{$block\_score \cdot height$} in order to ensure that no matter the number of empty slots to be removed from the chain trust score is less than the trust of a single filled slot the below function is used since it converges to \textit{$1$} at \textit{$\infty$} given \textit{$\lim_{t \to \infty} f(t) = 1$}. 
\bigbreak
\noindent
{\[ f(t) = \alpha \int_{0}^{t} e^{-\alpha x}dx \]}for \textit{$\alpha \in [0,\infty)$}.The value of \textit{$\alpha$} is tune-able in order to get the preferred behavior of the chain selection rule, given it represents the speed of convergence. 
\bigbreak
\noindent
The tune-able nature of \textit{$\alpha$} allows the chain selection rule presented here to be seen as a superset of some common preexisting chain selection rules. With \textit{$\alpha = 0$}, the chain selection rule will essentially prefer the chain with the most blocks acting as a traditional longest chain rule as is the Peercoin approach. When \textit{$\alpha \to \infty$}, the chain selection rule will act as Cardano's Ouroboros chain selection rule. 
\bigbreak
\noindent
\begin{SCfigure}[0.5][ht]
\caption{Two competing chains without checkpoints}
\includegraphics[width=0.6\textwidth]{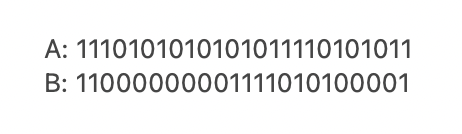}
\end{SCfigure}
\bigbreak
\noindent
In figure 1 two competing chains can be seen. Chain A will be selected by the longest chain rule in Pulsar given it is the denser chain.
\bigbreak
\noindent
\begin{SCfigure}[0.5][ht]
\caption{Two competing chains with checkpoints}
\includegraphics[width=0.6\textwidth]{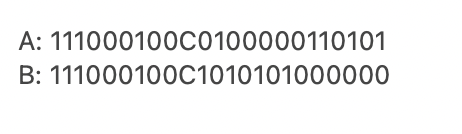}
\end{SCfigure}
\bigbreak
\noindent
In figure 2 competing chains with a checkpoint can be seen. Both chains obey the checkpoint and only the later chain is considered. A will win again as it has the denser chain. But a block in an empty slot on B could mean B is now the longest chain.
\bigbreak
\noindent
\begin{SCfigure}[0.5][ht]
\caption{Two competing chains with an ignored checkpoint}
\includegraphics[width=0.6\textwidth]{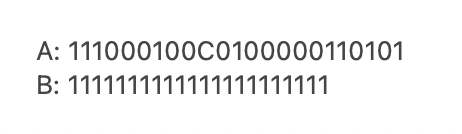}
\end{SCfigure}
\newpage
\noindent
In figure 3 a checkpoint exists but chain B ignores the checkpoint so despite the fact it has the denser chain it will not be selected. 
\bigbreak
\noindent
Although a formal proof for the chain selection rule in the partially-synchronous model is to follow it is worth briefly discussing several characteristics. First is that he function $W_b(t) = e^{-\alpha \cdot t}$ is monotonically decreasing for all $t \geq 0$, where $\alpha > 0$. Let's briefly prove the lemma.
\bigbreak
\noindent
Consider two values $t_1$ and $t_2$ where $0 \leq t_1 < t_2$. 
\\Since $\alpha > 0$, the exponential function $e^{-\alpha \cdot t}$ decreases as $t$ increases. 
\\Therefore, $e^{-\alpha \cdot t_1} > e^{-\alpha \cdot t_2}$. 
\\Additionally, the derivative of $W_b(t)$ with respect to $t$ is $\frac{dW_b}{dt} = -\alpha e^{-\alpha \cdot t}$. 
\\Since this derivative is negative for all $t \geq 0$, it confirms that $W_b(t)$ is decreasing at every point in its domain. 
\\Therefore, $W_b(t_1) = e^{-\alpha \cdot t_1} \geq e^{-\alpha \cdot t_2} = W_b(t_2)$, satisfying the criterion for monotonicity. 
\\Hence, we conclude that the function $W_b(t)$ is monotonically decreasing for all $t \geq 0$.
\[\forall t \geq 0, \frac{dW_{b}}{dt} < 0\]
\bigbreak
\noindent
We will next examine a simple game theoretic model of the chain selection rule, to attempt to find a nash equilibrium, assuming \textit{$\Delta = 0$} and a constant static block production time. A more complete analysis is to follow in a subsequent work.
\\
Since the honest network will try and build \textit{$b+1$} on \textit{$b$} the utility for that portion of the network can be simply computed as the probability of them producing a block. As such the utility is
{\[U_{h} = \phi_{h} \cdot r\]}
The malicious network is slightly more complicated as we must look at two situations, namely when the malicious network is able to produce a single block and when it is able to produce two blocks and replace the canonical chain, before the honest network produces a block.
{\[U_{m}^{1} = 0\]} since the block will not replace \textit{$b$} as it now has \textit{$2t$} for its empty slots. 
{\[U_{m}^{2} = \textit{$\phi_{m}^{2} \cdot 2r$}\]} since the malicious network must produce 2 blocks on their parallel chain to become the canonical chain and earn rewards at all. 
{\[U_{m}^{h} = \textit{$\phi_{m} \cdot r$}\]}if the malicious agent chooses to act honestly then their utility mirrors that of the honest network.
\bigbreak
\noindent
$\therefore $ if \textit{$\phi_{m} < \frac{1}{2}$} the utility is higher to build on the canonical chain.
\bigbreak
\noindent
This can be generalised to \textit{$\phi^{n} \cdot n \cdot r$} for \textit{$\{n \in \mathbb{Z} | 1 \leq n < h_{f}\}$}
For there to be a NE we need the condition {\[ u_{m}^{h} \geq u_{m}^{2}\]} to hold. That is that a malicious miner has an expected utility that is higher by acting honestly when they have the ability to produce a block than when they try and produce a competing chain. We assume that an honest miner will always act honestly.
\bigbreak
\noindent
$$1 - \phi_{h} \geq 2(1 - \phi_{h})^{2}$$
$$1 - \phi_{h} \geq 2(1 - 2\phi_{h} + \phi_{h}^{2})$$
$$0 \geq 2\phi_{h}^{2} - 3\phi_{h} + 1$$
$$0 \geq (\phi_{h} - 1)(2\phi_{h} - 1)$$
$$\frac{1}{2}\leq\phi_{h}\leq1$$
This implies a NE when \textit{$\phi_{h} \geq 0.5$} meaning that a malicious miner should always act honestly unless the total stake held by the malicious miner or their cartel is more than \textit{$0.5+\epsilon$}.
\section{Future work}
There are a number of active research areas that may form part of version 2 of Pulsar to improve upon the work outlined here. Slashing or a similar incentive system may be introduced to ensure correct network behavior; this could be by slashing the pool pledge or by modifying the reward paid to the pool when a block reward is earned. Slashing has several negatives which must be considered, not least of which is the complexity of correctly implementing slashing in a way that doesn't penalize people incorrectly. There is also the risk of reducing network participation, reducing decentralization, so directly slashing may not be ideal behavior. Alongside slashing the process of voting or co-signing blocks is an area of active research hence the desire for an aggregable signature scheme.
\bigbreak
\noindent
As mentioned before, checkpointing is an essential security feature in a proof of stake blockchain. Since Pulsar is designed to be a sidechain-focused consensus protocol, it could be possible to use the other chain to help checkpoint the network. This could be by checkpointing Pulsar blocks onto the other chain, say Bitcoin, or by linking Pulsar blocks to Bitcoin block hashes and ensuring linear ordering of Bitcoin block heights within Pulsar blocks.
\subsection{Need for punishments}
\label{sec:attack-model}
We model time in discrete slots $t=1,2,\dots$.  In each slot, a designated \emph{leader} $L_t$ is selected at random from the set of pools $P$, proportional to stake.  A valid \emph{block} in slot~$t$ is any header $B_t=(h_t,\sigma_t)$ where:
\begin{itemize}
  \item $h_t$ is the previous block hash,
  \item $\sigma_t=\mathrm{Sign}_{sk_{L_t}}(\mathrm{VRFSeed}\|h_t)$ is a signature under the leader's key,
  \item the VRF output $\mathrm{VRF}(sk_{L_t},h_t)$ falls below threshold $\tau\cdot p_{L_t}$.
\end{itemize}
An honest leader produces exactly one $B_t$ per slot.  An \emph{equivocating} leader may produce multiple distinct blocks $B_t^{(1)},\dots,B_t^{(k)}$ all satisfying validity.
\begin{lemma}
If leaders may equivocate without penalty, then the total number of valid forks grows exponentially in the number of slots.
\end{lemma}
\begin{proof}
Let $f_t$ denote the number of competing heads at slot $t$.  Base: $f_0=1$.  At slot~$t$, an equivocating leader can extend each of the $f_{t-1}$ heads with up to $k$ blocks, yielding
\[
f_t \ge f_{t-1} \cdot k
\]
By induction, $f_t \ge k^t$, so forks blow up exponentially.
\end{proof}
\begin{lemma}
If the next slot's leader colludes, then flooding $k$ variants with embedded transactions yields an expected MEV revenue proportional to $k$.
\end{lemma}
\begin{proof}
Given $k$ equally valid forks containing distinct profitable bundles, a colluding next leader selects the one with highest oracle‐measured reward.  Under mild regularity, the maximum of $k$ i.i.d. revenue draws grows as $\Theta(k)$, so profit scales linearly.
\end{proof}
\begin{lemma}
If the slot‑$t$ key $sk_t$ is destroyed immediately after signing, no pool can produce more than one valid signature in slot~$t$, and any further blocks signed with the same key in that slot are rejected.
\end{lemma}
\begin{proof}
By construction of a KES: upon signing message $m$, the key $sk_t$ is ratcheted to a fresh key $sk'_t$, rendering $sk_t$ unusable thereafter. Consequently, only one signature $\sigma_t=F_s(sk_t,m)$ can be generated per slot. Honest nodes also reject any blocks bearing a duplicated signature for slot~$t$, fully preventing equivocation.
\end{proof}
\begin{lemma}
Under an on‑chain slashing rule that forfeits a leader’s entire stake and all future block‑reward and MEV entitlements upon double‑signing, any attempt to equivocate is strictly unprofitable.
\end{lemma}
\begin{proof}
Let:
\[
  R_{\max}
  = \max_{\text{slot }t}\bigl(\text{MEV revenue in slot }t\bigr),
  \quad
  U_{\text{future}}
  = \sum_{u=t+1}^{\infty}\beta^{\,u-(t+1)}\bigl(r + R_{\max}\bigr),
\]
where \(r\) is the honest block reward per slot, and \(0<\beta<1\) is the leader’s discount factor.  If a leader equivocates in slot \(t\), they lose:
\[
  S_{\text{slash}}
  = \text{current stake} 
  \;+\;
  U_{\text{future}},
\]
but can gain at most \(R_{\max}\) in that slot.  Their net utility change is
\[
  U_{eq}
  = R_{\max} - S_{\text{slash}}
  < 0,
\]
since \(U_{\text{future}}\) dominates any single‑slot \(R_{\max}\).  Therefore, a rational, stake‑maximizing leader will never equivocate.
\end{proof}

\printbibliography
\end{document}